\documentclass[10pt, journal]{IEEEtran}

\ifCLASSINFOpdf
\else
\fi
\usepackage{multicol}
\usepackage{lipsum}
\usepackage{setspace}
\usepackage{amsmath}
\usepackage{cases}
\usepackage{cite}
\usepackage{amsfonts}
\usepackage{amssymb}
\usepackage[linesnumbered,ruled,vlined]{algorithm2e}

\usepackage{array}
\usepackage[margin=1.0in]{geometry}
\usepackage{amsthm}
\usepackage{empheq}
\renewcommand*{\stackrel}{%
\mathrel\bgroup\stack@relbin
}

\newtheorem{theorem}{Theorem}[section]
\newtheorem{lemma}[theorem]{Lemma}

\ifCLASSINFOpdf
   \usepackage[pdftex]{graphicx}
  \graphicspath{{../pdf/}{../jpeg/}{../eps/}}
  \DeclareGraphicsExtensions{.eps,.pdf,.jpeg,.png}
\else
   \usepackage[dvips]{graphicx}
   \graphicspath{{../eps/}}
   \DeclareGraphicsExtensions{.eps}
\fi
\usepackage{epsfig}
\usepackage{fmtcount}
\usepackage{subfigure}
\usepackage{array}
\usepackage{url}
\usepackage[linesnumbered,ruled,vlined]{algorithm2e}
\usepackage{multicol}
\usepackage{float}
\usepackage{lipsum}
\usepackage{color}

\usepackage{mathtools}  
\usepackage{tabulary}
\usepackage{booktabs}
\usepackage{stfloats}

\begin{document}%\fontsize{10}{12}\rm

%\begin{document}
%\title{Marginal Power Efficiency based Power Control for Wireless Multihop Networks}

\title{{\fontsize{18}{18}\selectfont Distributed Power Allocations in Heterogeneous Networks with Dual Connectivity using Backhaul State Information}}


\author{Syed~Amaar~Ahmad and Dinesh~Datla, \emph{Member, IEEE}
\thanks{Syed Amaar Ahmad and Dinesh Datla were affiliated with the Department of Electrical and Computer Engineering, Virginia Tech. Syed A. Ahmad is currently working with Datawiz Corporation and D. Datla is with Harris Corporation, Lynchburg, VA. Email: \{saahmad,ddatla\}@vt.edu.}}% <-this % stops a space
%\thanks{Manuscript received xxxx; revised xxxx}


\markboth{}%
{Shell \MakeLowercase{\textit{et al.}}: Bare Demo of IEEEtran.cls for Journals}

\maketitle

\begin{abstract}
LTE release 12 proposes the use of dual connectivity in heterogeneous cellular networks, where a user equipment (UE) maintains parallel connections to a macrocell base station and to a low-tier node such as a picocell base station or relay. In this paper, we propose a distributed multi-objective power control scheme where each UE independently adapts its transmit power on its dual connections, where the connections could possess unequal bandwidths and non-ideal backhaul links. In the proposed scheme, the UEs can dynamically switch their objectives between data rate maximization and transmit power minimization as the backhaul load varies. To address the coupling between interference and the backhaul load, we propose a low-overhead convergence mechanism which does not require explicit coordination between UEs and also derive a closed-form expression of the transmit power levels at equilibrium. Simulation results show that our scheme performs with higher aggregate end-to-end data rate and significant power saving in comparison to a scheme that employs a greedy algorithm and a scheme that employs only waterfilling.
\end{abstract}

\begin{IEEEkeywords}
Waterfilling, Energy-efficiency, HetNets, Small cells, Gaussian interference channel, Cross-layer design
\end{IEEEkeywords}

\IEEEpeerreviewmaketitle

\section{Introduction}
\IEEEPARstart{L}{TE} envisions the use of multi-tier access points in a cellular network to increase the coverage region of a base station \cite{Amitava2012}. A heterogeneous cellular network, consisting of a macrocell overlaid with small cells (e.g. picocells and relays), provides an efficient way for a cellular system to support the growing data rate demand. Smalls cells can alleviate the burden on a macrocell by offloading its users and the associated load \cite{Elsawy_2013}. Moreover, handheld devices are often equipped with multi-channel radio transceivers so as to enable multi-layer and multi-band connectivity (e.g. LTE, WiFi) \cite{Hu_wifi_lte_2012}. 

\begin{figure}[!t]
\begin{center}
        \centering\includegraphics[width=0.5\textwidth]{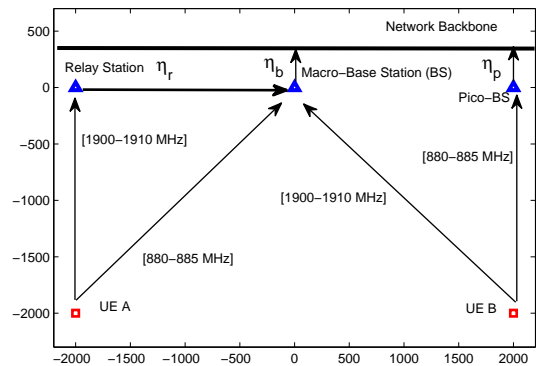}
\caption{A depiction of a heterogeneous network where UEs have dual connectivity over orthogonal channels with unequal bandwidth.}
\label{fig:furt1}
\end{center}
\end{figure}
LTE release 12 introduces enhancements to radio resource management in heterogeneous networks. One such enhancement, namely \emph{dual connectivity}, has been introduced to combat non-ideal backhaul links of base stations and small cells~\cite{3gppDCP12, LeanCarrier2013, Astely2013, Eric2013}. A large number of users can impose an excessive load on the backhaul links, notwithstanding the large backhaul capacity supported by the LTE architecture \cite{cisco2014backhaul, Zaki2013, abuali2013}. Under dual connectivity, a UE can utilize radio resources made available by more than two access points with limited backhaul capacity.   
%Release 12 introduces the notion of \emph{dual connectivity}, where a UE maintains its connection to the macro base station in addition to a parallel connection to a low-tier node. 

Dual connectivity can imply a variety of configurations such as enabling downlink and uplink connections on different tiers. As shown in Fig.~\ref{fig:furt1}, we consider dual connectivity where UEs have simultaneous uplink connections to two different access points~\cite{Eric2013}, each with limited (non-ideal) backhaul capacity. In this scenario, the focus of this paper is on distributed transmit power allocation on the two connections by UEs, where each UE makes its power allocation decisions independently of others. {This is in contrast to conventional link layer power allocations which are either agnostic to the state of the backhaul or aim to optimize some cross-layer objective function where nodes connect to a single access point only. In \cite{bambos2000}, nodes with point-to-point links adapt transmit power to achieve some minimum target SINR using the Foschini-Miljanic (FM) power control algorithm~\cite{foschini_93}. In \cite{rasti3} and \cite{rasti2}, the FM algorithm is combined with opportunistic power control that is proposed in~\cite{sung_05} to also be able to opportunitically improve data rates. The common factor in distributed algorithms, such as \cite{foschini_93, bambos2000, rasti3,rasti2,xiao_03}, is that nodes make autonomous decisions to maximize their individual local objective. In contrast, in cross-layer optimization problems, such as those presented in \cite{chiang05, Tran2012_JCPC, Tran2013_JCPC}, nodes adapt to maximize some centralized network-wide utility function through joint power and congestion control. Each node distributedly allocates power so that the total traffic load on any link does not exceed the available capacity. While distributed in nature these algorithms depend on message-passing between transmitters using a flooding protocol.}  

In our proposed adaptive scheme, UEs with dual connections autonomously allocate transmit power based on two factors, namely the current load in the backhaul links and channel link quality indicators, using feedback with a much lower overhead than the feedback used in \cite{chiang05, Tran2012_JCPC, Tran2013_JCPC}. Our scheme does so with the objective of achieving higher data rates and significantly reduced power consumption. A UE performs conventional waterfilling to maximize the data rate on its access links when its backhaul links have high capacity. Unlike \cite{waterfilling_Shum07}, {which studies waterfilling in Gaussian interference channels}, our approach also takes into consideration the impact of limited backhaul capacity. If the backhaul links become overloaded, the UE switches to a transmit power minimization mode and reduces its power (and data rate) on either or both its access links until the backhaul links are load-balanced. This energy-efficient strategy not only improves the UEs' battery life, it also helps alleviate congestion in the backhaul and decreases co-channel interference. Simulation results show that our scheme reduces the average transmit power of a UE to 40$\%$ as compared to a waterfilling scheme for the same data rate when the backhaul links are over-loaded. Given that UEs adapt autonomously, it is challenging to achieve stable system performance due to the coupling between interference and load on the backhaul links. We propose a mechanism by which individual power allocations by the UEs converge to an equilibrium point despite a lack of explicit coordination. We also derive closed-form expression for the converged transmit power levels.

We present the system model in Section II followed by the problem formulation in Section III. In Section IV, we propose the backhaul state-based adaptation scheme. In Section V, we explore convergence properties of the scheme followed by simulation results in Section VI. Finally, Section VII presents the conclusions of the paper.

\section{System Model}
We consider a cellular network that comprises a set of UEs, denoted by ${\mathcal{N}}=\{1,2,\cdots, n\}$, {which are located in a region served by a single macrocell base station (MBS)}. There are additional Points of Access (PoAs) that include $N_p$ picocell base stations (PBS) and $N_r$ relays (RS). A PBS could be treated as proxy for a Wi-Fi or femtocell access point. In the uplink, PBSs and RSs receive data from a UE and forward it to the MBS using a decode-and-forward scheme. The PoAs are denoted by a set, denoted by ${\mathcal{R}}=\{1,2,\cdots,N_r,\cdots, N_r+N_p, N_r+N_p+1\}$, where an RS $r$, a PBS $p$  and MBS $b$ are such that $r\leq N_r$, $N_r+1 \leq p \leq N_r+N_p$ and $b=N_r+N_p+1$. All PoAs share a common network backbone. %The dual links of the UEs are represented by the sets ${\mathcal{L}}_1 \subseteq {\mathcal{N}}$ and ${\mathcal{L}}_2 \subseteq {\mathcal{N}}$. 
The PoA of UE $i$ on its first access link is denoted as $a^{(1)}_i \in {\mathcal{R}}$ and the PoA of its second access link is $a^{(2)}_i \in {\mathcal{R}}$. 

All UEs can adapt their transmit power levels and, consequently, their data rate. The maximum available transmit power of the UEs is denoted by ${\bf{P}}_{\max}=\left[P_{\max,1},\cdots, P_{\max,n}\right]$ watts. The transmit powers of nodes on their respective first links are represented by the vector ${\bf{P}}_1=\left[P^{(1)}_{1}, \cdots, P^{(1)}_{n}\right]$ and those on the second links are represented by ${\bf{P}}_2=\left[P^{(2)}_{1}, \cdots, P^{(2)}_{n}\right]$, where $P^{(1)}_{i}+P^{(2)}_{i}\leq P_{\max,i}, \forall i \in {\mathcal{N}}$. 

The access links, denoted by two $n-$dimensional vectors ${\bf{L}}_1$ and ${\bf{L}}_2$, operate over a set of channels, denoted by ${\mathcal{F}}=\{1,2,\cdots, F\}$, whose respective bandwidths are represented by ${\mathcal{W}}=[W_1, W_2, \cdots, W_F]$. Channels may be re-used with the following restrictions. The access links of UE $i$ operate on orthogonal channels, i.e. $f^{(1)}_i \neq f^{(2)}_{i}$, and no two UEs transmit to a PoA using the same channel. %Thus, there may be mutual co-channel interference between access links in ${\bf{L}}_1$ and ${\bf{L}}_2$. 

The complex-valued channel gain between UE $i$'s transmitter and the PoA receiver of UE $j$ on channel $f^{(x)}_j\in {\mathcal{F}}$ is represented as $h^{(x)}_{ij}$, the corresponding channel power gain is given by $g^{(x)}_{ij}=|h^{(x)}_{ij}|^2$ and $n^{(x)}_i$ is the noise power, where $x\in \{1,2\}$ denotes either of the access links of UE $i$. Accordingly, the effective interference, SINR and achievable rate of link $x$ for UE $i$, are respectively given by~\cite{sung_05}
\begin{eqnarray} 
{E^{(x)}_{i}} &=& {\frac{n^{(x)}_i+\sum^{2}_{y=1}\sum_{\substack{\forall j\neq i, \\f^{(x)}_i=f^{(y)}_j}}g^{(y)}_{ji}P^{(y)}_{j}}{g^{(x)}_{ii}}},\label{effintf}\\
\gamma^{(x)}_{i} &=& \frac{P^{(x)}_{i}}{E^{(x)}_{i}},\\
\eta^{(x)}_{i} &=& W^{(x)}_{i} \log_{2}(1+\gamma^{(x)}_{i}),
\end{eqnarray}
where {$W^{(x)}_{i}$ is the bandwidth of channel $f^{(x)}_{i}$} and $x\in \{1,2\}$. Corresponding to the access links we define four $n \times n$ normalized cross-link gain matrices ${\bf{F}}_{xy}$ as given by {
\begin{equation} 
\label{gainmatrix1}
{\bf{F}}_{xy}(i,j) =
\begin{cases}
0 \mbox{, \hspace{0.5cm} if $i= j$ or $f^{(x)}_i \neq f^{(y)}_j$} \\
\frac{g^{(y)}_{ji}}{g^{(x)}_{ii}} \mbox{,\hspace{0.2cm} otherwise}
\end{cases}
\end{equation}
that capture the interference from UE $j$ on its access link $y \in \{1,2\}$ to PoA of UE $i$ on its access link $x \in \{1,2\}$}. We also define $n-$dimensional vectors ${\bf{D}}_1$ and ${\bf{D}}_{2}$ that represent normalized noise powers on the two access links of the UEs such that ${\bf{D}}_1(i)= n^{(1)}_i/g^{(1)}_{ii}$ and ${\bf{D}}_{2}(i)= n^{(2)}_i/g^{(2)}_{ii}$ for UE $i$ where noise levels are proportional to the respective channel bandwidth \cite{verdu_2002}.
\begin{figure*}[!b]
\normalsize
\hrulefill
\begin{equation}
\label{maxflow_het}
\begin{split}
\eta_N &= \min \left(\eta_b,\sum_{\forall i:a^{(1)}_i=b} \eta^{(1)}_i+ \sum_{\forall i:a^{(2)}_i=b}\eta^{(2)}_i +
\sum^{N_r}_{r=1}\min\left(\eta_r, \sum_{\forall i:a^{(1)}_i=r} \eta^{(1)}_i+ \sum_{\forall i:a^{(2)}_i=r}\eta^{(2)}_i\right) \right) \\
& +\sum^{N_r+N_p}_{p=N_r+1}\min\left(\eta_p, \sum_{\forall i:a^{(1)}_i=p} \eta^{(1)}_i+ \sum_{\forall i:a^{(2)}_i=p}\eta^{(2)}_i\right)
\end{split}
\end{equation}
\end{figure*} 
\section{Problem Formulation}
{The backhaul capacity at the MBS, each PBS and each RS is denoted as $\eta_{b}$, $\eta_{p}$ and $\eta_{r}$ bps, respectively. As shown in Fig.~\ref{fig:furt1}, the MBS and PBSs directly forward data to the backhaul network backbone. An RS forwards its data to the backbone via the MBS}. The aggregate end-to-end data rate (i.e. network's capacity), denoted as $\eta_N$, is defined in equation~\eqref{maxflow_het} using the Max-Flow Min-Cut Theorem \cite{sherali}. Next we define the \emph{rate differential} which indicates the difference between the backhaul capacities of PBS $p$, RS $r$ and MBS $b$ and their respective aggregate data rate demand as given by
\begin{eqnarray}\label{Vr_het1}
V_{p} &=& \eta_{p} - \sum_{\substack{\forall i:a^{(1)}_{i} = p}} \eta^{(1)}_{i} - \sum_{\substack{\forall i:a^{(2)}_{i} = p}} \eta^{(2)}_{i} ,\hspace{0.1in}\mbox{}\\
\nonumber V_{r} &=& \min\left(\eta_r, V^{+}_b\right) - \sum_{\substack{\forall i:a^{(1)}_{i}=r}} \eta^{(1)}_{i} \\
\label{Vr_het2}
       &\;& - \sum_{\substack{\forall i:a^{(2)}_{i}=r}} \eta^{(2)}_{i}, \hspace{1.2in}\mbox{} \\ 
\label{Vr_het3}
V_{b} &=& \eta_b - \sum_{\substack{\forall i:a^{(1)}_{i} = b}} \eta^{(1)}_{i} - \sum_{\substack{\forall i:a^{(2)}_{i} = b}} \eta^{(2)}_{i} - \Gamma, \hspace{0in}\mbox{}
\end{eqnarray}
where
\begin{eqnarray}
\nonumber
\Gamma = \sum^{N_r}_{r=1}\min\left(\eta_r, \sum_{\substack{\forall i:a^{(1)}_i=r}} \eta^{(1)}_i+\sum_{\substack{\forall i:a^{(2)}_i=r}}\eta^{(2)}_i\right).     
\end{eqnarray}
The rate differential of an RS depends on the backhaul capacity of the MBS in equation~\eqref{Vr_het2}. 

We denote the rate differential of the PoA $a^{(x)}_i$ associated with the access link $x$ of UE $i$ as $V^{(x)}_i$. Whenever $V^{(x)}_i < 0$, the backhaul link at the PoA represents a \emph{bottleneck} link that limits the end-to-end data rate of UE $i$. {In this case the UE, which has a non-bottleneck access link,} cannot improve its data rate by increasing transmit power on that link and instead it may switch to power minimization. Thus, $V^{(x)^{+}}_i=\min\left(V^{(x)}_i,0\right)$ denotes the achievable rate improvement at a PoA for UE $i$. 

A UE employs waterfilling as the optimal allocation strategy to maximize data rate when the backhaul capacities of both its access links are sufficiently high \cite{perez05, Ahmadtop2015}.
\begin{theorem}
When the backhaul capacities is high enough, UE $i$ can maximize its data rate using a waterfilling power allocation such that \\
$P^{(1)^*}_{i} = \min \left(P_{\max,i},\frac{\left(W^{(1)}_i P_{\max,i}-W^{(2)}_{i}E^{(1)}_{i}+W^{(1)}_{i}E^{(2)}_{i}\right)^{+}}{W^{(1)}_{i}+W^{(2)}_i}\right) $ and \\
$P^{(2)^*}_{i} = \min \left(P_{\max,i},\frac{\left(W^{(2)}_i P_{\max,i}-W^{(1)}_{i}E^{(2)}_{i}+W^{(2)}_{i}E^{(1)}_{i}\right)^{+}}{W^{(1)}_{i}+W^{(2)}_i}\right)$. \label{WF_het} 
\end{theorem}
\begin{proof}
See \cite[p.~12]{Ahmadtop2015} for details.
\end{proof}
If only one PoA of a UE has a high enough backhaul capacity then the node can reduce the transmit power on its bottleneck access link and re-allocate the power to the other access link. In either situation the node allocates all of its available transmit power. If the backhaul capacities on both access links of a UE are limited, it is inefficient to expend all of its transmit power. Instead, the UE can aim to achieve the maximum end-to-end data rate by expending minimal transmit power. We formulate this as a multi-objective optimization problem for UE $i$ as follows
\begin{eqnarray} 
\begin{aligned}
\label{opt_func_het}
& \underset{P^{(1)}_{i}, P^{(2)}_{i}}{\text{maximize}}
& & -\left(P^{(1)}_{i}+P^{(2)}_{i}\right) \\
& \text{subject to}
& & P^{(1)}_{i}+P^{(2)}_{i} \leq {P_{\max,i}}, \\
&&& P^{(1)}_{i},P^{(2)}_{i} \geq 0, \\
&&& \psi\left(P^{(1)}_{i}, P^{(2)}_{i}\right)\geq R_{i}^{*} ,\\
\text{where} \:R_{i}^{*} = \: &  \underset{P^{(1)}_{i}, P^{(2)}_{i}}{\text{maximize}}
& & \psi\left(P^{(1)}_{i},P^{(2)}_{i}\right) \\
& \text{subject to}
& & P^{(1)}_{i}+P^{(2)}_{i} \leq {P_{\max,i}}, \\
&&& P^{(1)}_{i},P^{(2)}_{i} \geq 0
\end{aligned}
\end{eqnarray}
To solve \eqref{opt_func_het}, the first step involves computing $R^{*}_{i}$ where the end-to-end data rate improvement of UE $i$ is given by
\begin{eqnarray}
\nonumber
\psi(P^{(1)}_i, P^{(2)}_i) = \min\left(V^{(1)^{+}}_{i},\eta^{(1)}_{i} \right)+ \min\left(V^{(2)^{+}}_{i}, \eta^{(2)}_{i}\right).
\end{eqnarray} 
The second step involves computing the minimum transmit power allocation that can achieve the rate improvement $R^*_i$. Note that the power allocation in Theorem~\ref{WF_het} is a special case of \eqref{opt_func_het} when $V^{(1)^{+}}_{i}$ and  $V^{(2)^{+}}_{i}$ are large enough.

\section{Backhaul state-based Distributed Transmission}
We assume that adaptations occur in time intervals denoted as $k \in \{1,2,..\}$. In interval $k$, the instantaneous effective interference levels $E^{(x)}_{i}(k)$ and the rate differentials $V^{(x)}_i(k)$ are made available to UE $i$. {Optimization~\eqref{opt_func_het} may be solved using a greedy algorithm, where each UE makes a locally optimal decision of its transmit power allocations. However, due to inter-dependence between the interference levels and the rate differentials the greedy approach may result in an erratic and unstable data rate.} 

We propose a heuristic scheme called Backhaul state-based Distributed Transmission (BDT) which results in a significantly better performance as shown later. We further assume that a positive-valued constant $\tau$, known to all UEs, is a rate differential threshold that represents some tolerable load at the backhaul side of the PoAs. 
%As we will later show, this also helps enable the system to converge. 
\begin{table}[t]\caption{Backhaul states for UE $i$.}
\begin{center}% used the environment to augment the vertical space
% between the caption and the table
\begin{tabular}{r c p{0.01cm} }
\toprule
State & Rate Differentials\\
${\mathcal{S}}(1)$ & $ V^{(1)}_i(k)\geq 0 $ and $V^{(2)}_i(k)\geq 0 $ \\ 
${\mathcal{S}}(2)$ &$-\tau\leq V^{(1)}_i(k)<0 $ and $V^{(2)}_i(k)\geq 0 $ \\
${\mathcal{S}}(3)$ & $V^{(1)}_i(k)\geq 0 $ and $-\tau\leq V^{(2)}_i(k)<0 $\\
${\mathcal{S}}(4)$ & $-\tau\leq V^{(1)}_i(k)<0 $ and $-\tau\leq V^{(2)}_i(k)<0$\\
${\mathcal{S}}(5)$ & $ V^{(1)}_i(k)\geq 0 $ and $V^{(2)}_i(k)< -\tau$ \\
${\mathcal{S}}(6)$ & $ V^{(1)}_i(k)<-\tau $ and $V^{(2)}_i(k)\geq 0$\\
${\mathcal{S}}(7)$ & $-\tau\leq V^{(1)}_i(k)<0 $ and $V^{(2)}_i(k)< -\tau$\\
${\mathcal{S}}(8)$ & $V^{(1)}_i(k)<-\tau $ and $-\tau \leq V^{(2)}_i(k)< 0$\\
${\mathcal{S}}(9)$ & $V^{(1)}_i(k)<-\tau $ and $V^{(2)}_i(k)< -\tau$\\
\bottomrule
\end{tabular}
\end{center}
\label{tab:TableOfNotationForMyResearch}
\end{table}
In the proposed scheme, the following three backhaul states are encountered at any POA. When $V^{(x)}_i(k)\geq 0$, UE $i$ can improve its data rate by increasing its transmit power. When $V^{(x)}_i(k)< -\tau$, the PoA is overloaded and the UE is inefficiently using transmit power. In the intermediate case, namely $ -\tau \leq V^{(x)}_i(k)<0$, the UE can maintain its transmit power (and data rate) since the rate differential is tolerable. With dual connectivity, there are nine possible states of the backhaul links at the two PoAs of each UE as shown in Table~\ref{tab:TableOfNotationForMyResearch}. The UE can adapt to the states in one of the following ways: (i) waterfilling, (ii) maintaining transmit power, (iii) reducing its transmit power, and (iv) re-allocating transmit power between its access links. Accordingly, our proposed transmit power allocation algorithm is given by
\begin{equation} 
\label{power_alloc_het}
P^{(1)}_i(k+1) =
\begin{cases}
P^{(1)*}_i(k) \mbox{, if ${\mathcal{S}}(1)$} \\
 P^{(1)}_i(k) \mbox{, if ${\mathcal{S}}(2), {\mathcal{S}}(4), {\mathcal{S}}(7)$} \\
 P_{\max, i}-P^{(2)}_i(k+1) \mbox{, if ${\mathcal{S}}(3), {\mathcal{S}}(5)$} \\
ZP^{(1)}_{i}(k) \mbox{, if ${\mathcal{S}}(6), {\mathcal{S}}(8), {\mathcal{S}}(9)$} 
\end{cases}
\end{equation}
\begin{equation} 
\label{power_alloc_het2}
P^{(2)}_i(k+1) =
\begin{cases}
P^{(2)*}_i(k) \mbox{, if ${\mathcal{S}}(1)$}\\
P^{(2)}_i(k) \mbox{, if ${\mathcal{S}}(3),{\mathcal{S}}(4), {\mathcal{S}}(8)$} \\
P_{\max, i}-P^{(1)}_i(k+1)\mbox{, if ${\mathcal{S}}(2), {\mathcal{S}}(6)$} \\
ZP^{(2)}_{i}(k) \mbox{, if ${\mathcal{S}}(5), {\mathcal{S}}(7), {\mathcal{S}}(9)$}  
\end{cases}
\end{equation}
In state ${\mathcal{S}}(1)$, the UE maximizes its data rate using waterfilling allocation, where $P^{(1)*}_i(k)$ and $P^{(2)*}_i(k)$ are computed using Theorem~\ref{WF_het}. In states ${\mathcal{S}}(2)$ and ${\mathcal{S}}(3)$, the backhaul capacity for one access link is high enough and the load is within the tolerable limit on the other. In states ${\mathcal{S}}(5)$ and ${\mathcal{S}}(6)$, only one access point's backhaul link suffers overloading. In this case, the UE adapts by reducing its transmit power on the link with an overloaded backhaul by a constant factor $Z: 0<Z<1$ and re-allocating it to the other access link which has a high enough backhaul capacity. In states ${\mathcal{S}}(7)$, ${\mathcal{S}}(8)$ and ${\mathcal{S}}(9)$, the UE reduces transmit power on either or both access links. This continues until either the transmit power diminishes to zero or the associated overloading level drops to an acceptable level. In state ${\mathcal{S}}(4)$, where the load on both backhaul links is within the acceptable range, the UE maintains its transmit power. 

In a practical implementation of BDT, a 2-bit feedback on the backhaul state at a PoA can be used to achieve a low overhead. Moreover, since $V^{(x)}_i(k)$ is computed for an interval spanning several LTE resource blocks comprising several hundred bits the feedback can be piggybacked on the control channel feedback that is already provisioned in the standard~\cite{3gppDCP12}, thereby imposing negligible overhead. The distributed BDT adaptations converge when this feedback is provided to the UEs at every power control iteration as illustrated next. 
  
\section{Convergence} 
{BDT differs from a greedy algorithm implementation of equation~\eqref{opt_func_het} in the following ways: (i) when $V^{(x)^+}_i=0$, the UE  decreases its transmit power iteratively instead of immediately setting it to zero, and (ii) BDT aims to achieve $V^{(x)}_i\geq -\tau$ to avoid under-utilizing the backhaul capacities.} These features allow distributed allocations by the nodes to converge whereas a greedy algorithm may not result in stable system performance. A vector representation of the effective interference (equation~\eqref{effintf}) in iteration $k$ for all $n$ UEs at their two access links is given by{
\begin{equation} \label{effintfvec}
\begin{split}
{\bf{E}}_1(k) &= {\bf{D}}_1 + {\bf{F}}_{11}{\bf{P}}_1(k)+ {\bf{F}}_{21}{\bf{P}}_2(k) \: \text{,}\\
{\bf{E}}_{2}(k) &= {\bf{D}}_{2} + {\bf{F}}_{22}{\bf{P}}_2(k)+ {\bf{F}}_{12}{\bf{P}}_1(k), 
\end{split}
\end{equation}}
respectively. We define ${\bf{W}}_1=\mbox{Diag}[W^{(1)}_1, \cdots, W^{(1)}_n]$ and ${\bf{W}}_2=\mbox{Diag}[W^{(2)}_1, \cdots, W^{(2)}_n]$. Next we define 
\begin{equation}
\begin{split}
{\bf{\Lambda}}&=[{\bf{W}}_1+{\bf{W}}_2]^{-1}\\
&= \mbox{Diag}\left[\frac{1}{W^{(1)}_1+W^{(2)}_1}, \cdots, \frac{1}{W^{(1)}_n+W^{(2)}_n}\right].
\end{split}
\end{equation}
\subsection{High backhaul capacity regime}
The convergence of BDT when the backhaul capacity at each PoA is sufficiently high is discussed next.
\begin{theorem}
When BDT is used to solve equation~\eqref{opt_func_het}, the transmit powers converge to \\  {${\bf{P}}_1^*=[{\bf{I}}-{\bf{\Lambda}}{[}{\bf{W}}_2({\bf{F}}_{21}-{\bf{F}}_{11})+{\bf{W}}_1({\bf{F}}_{12}-{\bf{F}}_{22}){]}]^{-1}\\
{\bf{\Lambda}}{[}{\bf{P}}_{\max}-{\bf{W}}_2{\bf{D}}_{1}+{\bf{W}}_1{\bf{D}}_{2}+ ({\bf{W}}_1{\bf{F}}_{22}-{\bf{W}}_2{\bf{F}}_{21}){\bf{P}}_{\max} {]}$ and ${\bf{P}}_2^*= {\bf{P}}_{\max}-{\bf{P}}_1^*$} , given that ${\bf{0}}<{\bf{P}}_1^*< {\bf{P}}_{\max}$ and each PoA has a large enough backhaul capacity. \label{theoremx_het}
\end{theorem}
\begin{proof}
{When backhaul capacity at each PoA is sufficiently high, each UE allocates transmit power using waterfilling (see the adaptation rule for $\mathcal{S}(1)$ in equations~\eqref{power_alloc_het} and \eqref{power_alloc_het2}), where the transmit power allocations on the dual connections are related as ${\bf{P}}_{2}(k)={\bf{P}}_{\max}-{\bf{P}}_{1}(k)$. Given that ${\bf{0}}<{\bf{P}}_1^*< {\bf{P}}_{\max}$, the power updates for state $\mathcal{S}(1)$ in equation~\eqref{power_alloc_het} can be represented as a linear system and elaborated using equation~\eqref{effintfvec} as given by 
\begin{align}
\nonumber
{\bf{P}}_1(k+1)&={\bf{\Lambda}}\left[{\bf{W}}_1{\bf{P}}_{\max}-{\bf{W}}_2{\bf{E}}_1(k)+{\bf{W}}_1{\bf{E}}_{2}(k)\right] \\
\nonumber
&={\bf{\Lambda}}\left[ {\bf{W}}_{1}{\bf{P}}_{\max} - {\bf{W}}_{2}({\bf{D}}_{1}+{\bf{F}}_{11}{\bf{P}}_{1}(k)+ \right.\\
\nonumber
&\hphantom{={\bf{\Lambda}}\:\:}\left. {\bf{F}}_{21}({\bf{P}}_{\max}-{\bf{P}}_{1}(k)) +{\bf{W}}_2({\bf{D}}_{2} + \right.\\
\nonumber
&\hphantom{={\bf{\Lambda}}\:\:} \left. {\bf{F}}_{22}({\bf{P}}_{\max} - {\bf{P}}_1(k) + {\bf{F}}_{21}{\bf{P}}_{1}(k))) \right] \\
\label{update_ue_rs}
&={\bf{N}}+{\bf{M}}{\bf{P}}_1(k),
\end{align}
where matrices ${\bf{N}}$ and ${\bf{M}}$ are defined as given by
\begin{align}
\nonumber
{\bf{N}}&= {\bf{\Lambda}}\left[ {\bf{P}}_{\max}-{\bf{W}}_2{\bf{D}}_{1}+{\bf{W}}_1{\bf{D}}_{2}+ \right. \\           
         &  \hphantom{={\bf{\Lambda}}\:\:} \left.   ({\bf{W}}_1{\bf{F}}_{22}-{\bf{W}}_2{\bf{F}}_{21}){\bf{P}}_{\max} \right] \\
{\bf{M}}&= {\bf{\Lambda}}[{\bf{W}}_2({\bf{F}}_{21}-{\bf{F}}_{11})+{\bf{W}}_1({\bf{F}}_{12}-{\bf{F}}_{22})].
\end{align}
By definition \cite[p.~618]{meyer04}, if the spectral radius of matrix $\bf{M}$ is less than one the linear system in equation~\eqref{update_ue_rs} evolves to a fixed point which can be derived as given by} 
\begin{align}
\nonumber
{\bf{P}}_1(k)&={\bf{N}}+{\bf{M}}\left({\bf{N}} + {\bf{M}}\left({\bf{N}}+{\bf{M}}\left(\cdots {\bf{P}_1}(0) \right) \right)\right) \\
\nonumber
{\bf{P}}^*_1&=\lim\limits_{k \rightarrow \infty} {\bf{P}}_1(k)=\left[{\bf{I-M}}\right]^{-1}{\bf{N}}\\
\nonumber
&=[{\bf{I}}-{\bf{\Lambda}}{[}{\bf{W}}_2({\bf{F}}_{21}-{\bf{F}}_{11})+{\bf{W}}_1({\bf{F}}_{12}-{\bf{F}}_{22}){]}]^{-1}\\
\nonumber
& \hphantom{=\:\:\:} {\bf{\Lambda}}\left[ {\bf{P}}_{\max}-{\bf{W}}_2{\bf{D}}_{1}+{\bf{W}}_1{\bf{D}}_{2}+ \right. \\
\label{powerseries}
&  \hphantom{={\bf{\Lambda}}\:\:} \left.   ({\bf{W}}_1{\bf{F}}_{22}-{\bf{W}}_2{\bf{F}}_{21}){\bf{P}}_{\max} \right] \; \\
\text{and} \; {\bf{P}}_2^* &= \lim\limits_{k \rightarrow \infty} {\bf{P}}_2(k)={\bf{P}}_{\max}- {\bf{P}}^*_1 ,  
\end{align}
where ${\bf{P}}_1(0)$ is the initial transmit power vector of the first links and the corresponding transmit powers on the second access links are simply the difference between ${\bf{P}}_{\max}$ and ${\bf{P}}_1(k+1)$ from equation~\eqref{update_ue_rs}.
\end{proof}
\subsection{Limited backhaul capacity regime}
Next, we consider the general case where the backhaul capacity at the PoAs is limited. {The spectral radius operator is denoted by $\rho(\cdot)$. Note that, whenever $V^{(x)}_i < 0$ the backhaul link at the PoA represents a \emph{bottleneck} link while the access link of UE $i$ is considered a \emph{non-bottleneck} link.}
\begin{lemma}
Under BDT, if the spectral radius $\rho\left({\bf{M}}\right)<1$ then the SINR of a non-bottleneck access link changes such that $\gamma^{(x)}_i(k+2)> Z^2 \gamma^{(x)}_i(k)$. \label{nbl_SINR_iter}
\end{lemma}
\begin{proof}
Consider a scenario where the transmit power levels are at equilibrium as per Theorem~\ref{theoremx_het}. Now, assume that a UE $i$ with a non-bottleneck link re-scales its transmit power once by $Z$ and a co-channel UE $j$ with a bottleneck access link uses waterfilling allocation to exploit interference reduction and increase transmit power. Due to the constraint on the spectral radius there is a limit on how much interference may then increase on the non-bottleneck access link of UE $i$. The implied sequence of transmit power updates is as follows: $P^{(x)}_i(k+1)= Z P^{(x)}_i(k)$ $\rightarrow$ $E^{(x)}_j(k+1)> Z E^{(x)}_j(k)$ $\rightarrow$
$P^{(x)}_j(k+2) < Z^{-1} P^{(x)}_j(k) $ $\rightarrow$ $E^{(x)}_i(k+2) < Z^{-1} E^{(x)}_i(k) $. Thus, the SINR is given by $\gamma^{(x)}_i(k+2) \left(= \frac{P^{(x)}_i(k+2)}{E^{(x)}_i(k+2)} \right) > Z^2 \gamma^{(x)}_i(k)$. 
\end{proof}
Lemma~\ref{nbl_SINR_iter} indicates that the possible change in the rate of a non-bottleneck link between any 2 consecutive iterations is bounded. 

\begin{theorem}%$${\bf{E}}_{\max}\leq {\bf{P}}_{\max}$
For a fixed rate differential threshold $\tau$, BDT adaptations converge for some value of transmit power scaling factor $0<Z<1$ if the spectral radius $\rho\left({\bf{M}}\right)<1$.\label{finiteConvf_net}
\end{theorem}
\begin{proof}
Suppose we set some high enough rate differential threshold with value $\tau^*$ such that the nodes do not cycle back and forth between rescaling their transmit power by $Z$ and waterfilling allocation. As each UE now operates in states ${\mathcal{S}}(1)$, ${\mathcal{S}}(2)$, ${\mathcal{S}}(3)$ and $\mathcal{S}(4)$, its power allocation will only follow the adaptation rules associated with these states in equations~\eqref{power_alloc_het} and \eqref{power_alloc_het2}. Thus, the evolution of transmit power for the system is equivalent to that in equation~\eqref{powerseries} except that some UEs may not update their transmit power every iteration (i.e. their transmit power level is based on an update from an earlier iteration). This as an \emph{asynchronous iterative} system \cite{frommera2000, Amaarasync2013,waterfilling_Shum07}. It is known that such a system converges if the iterative matrix ${\bf{M}}$ in equation~\eqref{powerseries} has a spectral radius less than one. 

In general, setting a high $\tau$ ensures that the UEs do not switch back and forth between waterfilling and transmit power reduction as any load level is tolerated. An alternative approach is to fix the threshold $\tau$ and instead set the transmit power scaling factor $Z$ in equations~\eqref{power_alloc_het} and \eqref{power_alloc_het2} to some value so that the nodes do not oscillate between different states. 

Now, suppose that at equilibrium we reset the rate differential threshold to some $\tau^{\dagger}<\tau^*$. This might induce the UEs to reduce their transmit power and data rates to bring back $-\tau^{\dagger}\leq V^{(x)}_i(k)<0$. For some $Z$ close to one, the corresponding change in the SINRs and thus the rate differentials between 2 consecutive iterations is such that $V^{(x)}_i(k)-V^{(x)}_i(k+2)<\tau^{\dagger}$ at all PoAs as per Lemma~\ref{nbl_SINR_iter}. Thus, the set of bottleneck links and non-bottleneck links will not change. At this point the transmit power levels on the non-bottleneck links will be maintained. Moreover, since $\rho\left({\bf{M}}\right)<1$ the transmit power of bottleneck links will also asynchronously converge.
\end{proof}
Theorem~\ref{finiteConvf_net} and the ensuing discussion indicates that local adaptations can converge for any $\tau$ albeit at the cost of slow convergence if $Z$ is closer to one. {The convergence results that hold in Theorems V.1 and V.3 apply to the worst case scenario when every UE is transmitting. However, uplink traffic in practical wireless systems is typically bursty. In such situations, there will be less instantaneous interference as some UEs may not be momentarily transmitting. Thus, if the system performance converges in the worst case it will also converge with bursty traffic.} 

\begin{figure*}[!t]
\begin{center}
\subfigure[]{\includegraphics[width=0.48\textwidth]{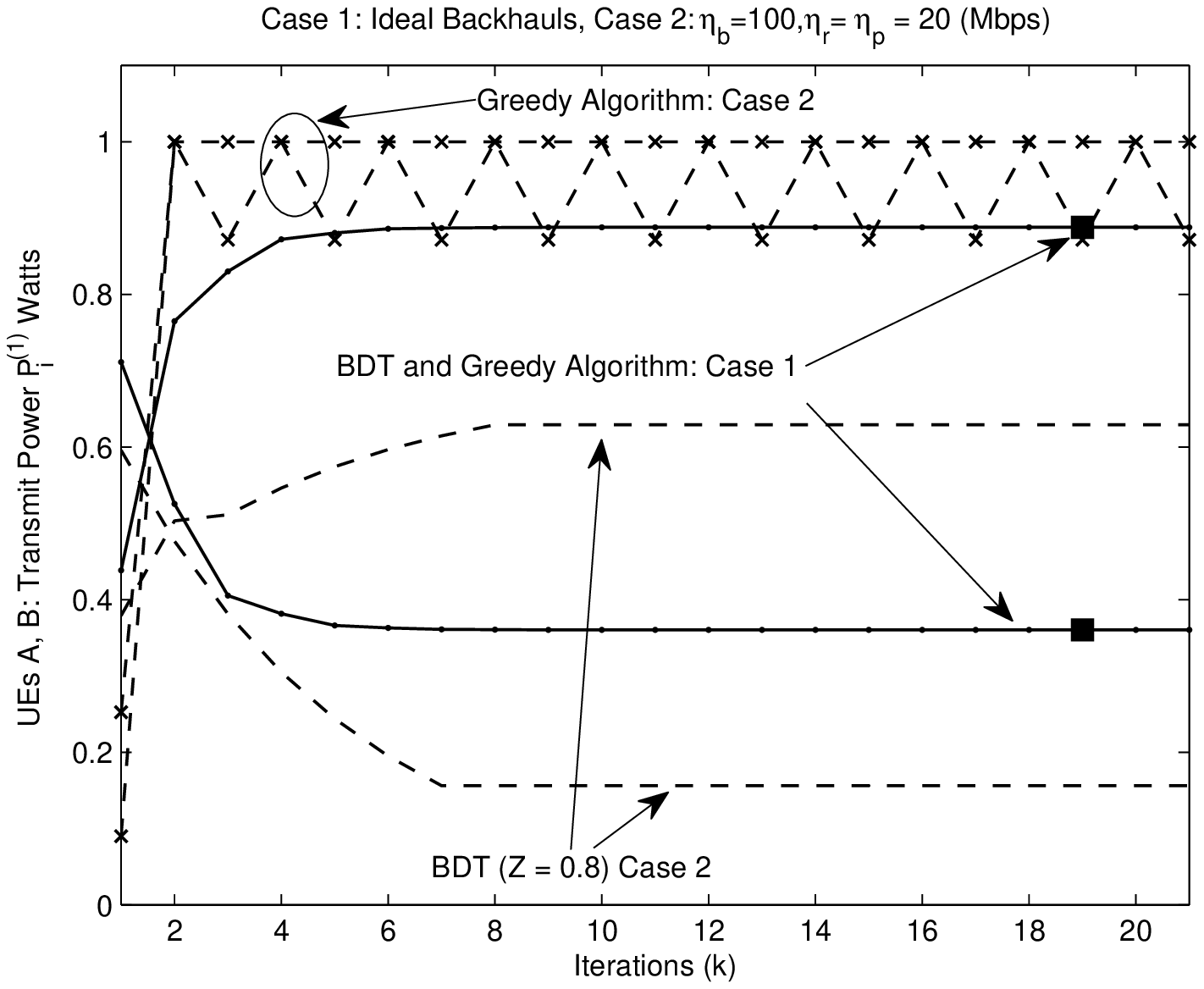}}
\subfigure[]{\includegraphics[width=0.48\textwidth]{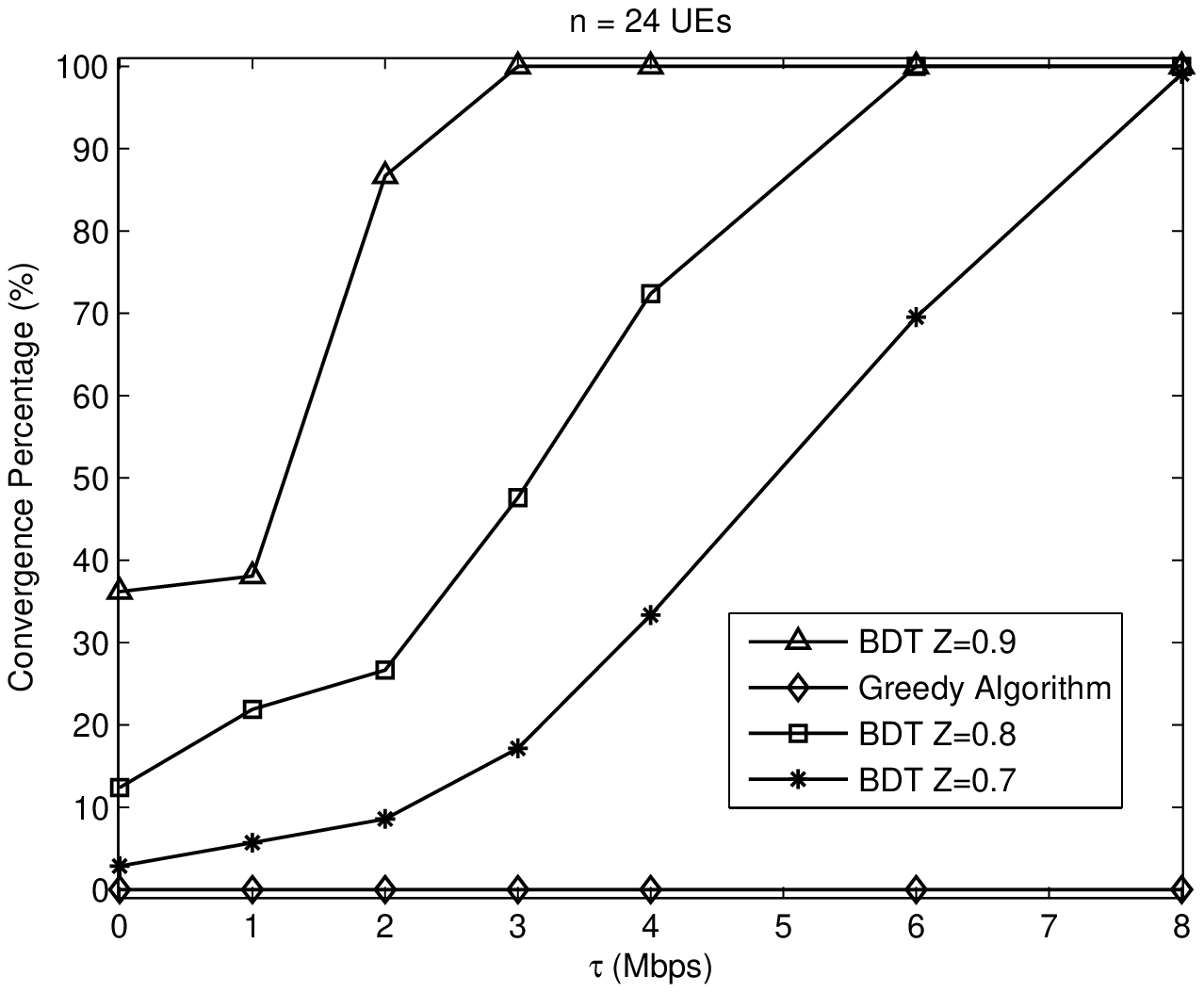}}
\caption{(a) Convergence to fixed point shown in equation~\eqref{powerseries} in the high backhaul regime (only $P^{(1)}_i$ is shown). {(b) Convergence percentage for different rate differential thresholds. We see guaranteed convergence for BDT as per Theorem~\ref{finiteConvf_net} whereas the greedy algorithm has poor convergence properties}.}
\label{fig:convg}
\end{center}
\end{figure*}

\subsection{Fixed Target SINR}
In cellular networks, target SINR is often an important criterion to meet for traditional power control applications \cite{bambos2000,foschini_93,rasti_10}. Consider a network which comprises two sets of UEs: one that has dual connectivity and adapts using BDT and the other that comprises traditional single-link nodes with only one PoA and a fixed target SINR. We show that the distributed power adaptations in the system reach a stable allocation. Let ${\bf{Q}}=\mbox{Diag}[q_1, q_2, \cdots, q_n]$ be a quasi-identity matrix where $q_i=1$ (single-PoA) and $q_i=0$ (UE with dual connectivity). For ease of analysis, we assume that the UEs with a single access point with fixed target SINRs interfere only with the access links ${\bf{L}}_1$ of the UEs with dual connectivity. We can thus modify equation~\eqref{update_ue_rs} as given by
\begin{align}
%\begin{split}
\nonumber
{\bf{P}}_1(k+1)&={\overline{\bf{Q}}}[{\bf{N}}+{\bf{M}}{\bf{P}}_1(k)] + {\bf{Q}}{\bf{B}}[{\bf{D}}_1+{\bf{F}}_{11}{\bf{P}}_1(k)]\\ 
\label{coexist}
&={\overline{\bf{Q}}}{\bf{N}}+{{\bf{Q}}}{\bf{D}}_1+[{\bf{Q}}{\bf{B}}{\bf{F}}_{11}-{\overline{\bf{Q}}}{\bf{M}}]{\bf{P}}_1(k),
%\end{split}
\end{align}
where ${\overline{\bf{Q}}}={\bf{I}}-{\bf{Q}}$, ${\bf{B}}=\mbox{Diag}[\beta_1, \beta_2, \cdots \beta_n]$ and $\beta_i$ is the fixed target SINR of UE $i$ ($\beta_i=0$ if $q_i=0$). Since equation~\eqref{coexist} is also a linear system, the convergence results in Theorems~\ref{theoremx_het} and~\ref{finiteConvf_net} will apply if the spectral radius of $[{\bf{Q}}{\bf{B}}{\bf{F}}_1-{\overline{\bf{Q}}}{\bf{M}}]$ is less than one. Thus, the power allocations will again converge for both sets of nodes.

\section{Simulation Results}
In this section, we use Matlab-based simulations of a heterogeneous cellular network to evaluate the performance of our scheme. We assume $n_o=-190$ dBW/Hz as the noise power spectral density and $P_{\max,i}=1.0$ watts for all UEs. We assume that the cross-link gains ($g^{(f)}_{i, r_i}, g^{(q)}_{i, b}$, etc.) are of the form $\kappa^{(f)}_{i, a}\cdot d_{i, a}^{-\alpha}$, where $d_{i,a}$ is distance between UE $i$ and PoA $a\in \mathcal{R}$, $\alpha$ is the path loss exponent and $\kappa^{(f)}_{i, a}$ is an exponentially distributed random variable with unit variance on channel $f$ due to Rayleigh fading. The fading gains are assumed to be independent and identically distributed for all channels. {The picocell base stations and relays are placed randomly within a $3$ km $ \times$ $3.2$ km area centered around an MBS in $(N_r+N_p+1)$ equal non-overlapping rectangular regions. The $n$ UEs are then placed within circular regions of radii $R_L$ m around these PoAs with uniform distribution}. Unless otherwise stated, we assume the following in each simulation trial: $R_L=200$ m, $N_r=3$, $\eta_r=100$ Mbps, $N_p=4$, $\eta_p=200$ Mbps, $\eta_b=1$ Gbps, path loss exponent $\alpha=3.7$, channel bandwidths are set such that $W^{(1)}_i, W^{(2)}_i\in \{1,5\}$ MHz $\forall i \in {\mathcal{N}}$, $\tau=5$ Mbps and $Z=0.9$ for BDT. The simulation trials involve $50$ power control iterations. 

{We compare the performance of BDT with that achieved when each UE employs either a greedy algorithm to implement equation~\eqref{opt_func_het} or waterfilling allocation (WF) on its two access links \cite{goldsmith05}. {To reiterate, under the greedy approach each UE makes an optimal power allocation based on the instantaneous channel and backhaul states}. In the access links ${\bf{L}}_1$, all UEs connect to the closest PoA (an RS or a PBS other than the MBS) whereas the UEs connect to the MBS on their access links in ${\bf{L}}_2$.}  

\subsection{Convergence Example}
We first consider the network topology as shown in Fig.~\ref{fig:furt1} where the grid coordinates (in km) of the nodes are as follows: MBS $(0,0)$, PBS $(2,0)$, RS $(-2,0)$, UE A $(-2,-2)$ and UE B $(2,-2)$. Under this topology we get 
${\bf{F}}_{12}=\begin{bmatrix}
       0 & 1        \\[0.3em]
       0.0509 & 0        \\[0.3em]
     \end{bmatrix}$, ${\bf{F}}_{21}=\begin{bmatrix}
       0 & 0.5        \\[0.3em]
       0.0509 & 0        \\[0.3em]
     \end{bmatrix}$ and {${\bf{F}}_{11}{=}{\bf{F}}_{22}=\begin{bmatrix}
       0 & 0        \\[0.3em]
       0 & 0        \\[0.3em]
     \end{bmatrix}$}, where the fading gain on the link between UE B and MBS is $0.5$ and the fading gains on all other links are set to $1$. Similarly, we have ${\bf{D}}_1=[0.0164 \hspace{0.05in} 0.059]^{T}$, ${\bf{D}}_1=[0.0295 \hspace{0.05in} 0.0082]^{T}$, $W^{(1)}_i=10$ MHz and $W^{(2)}_i=5$ MHz. As shown in Fig.~\ref{fig:convg}a, the transmit power levels of both BDT and greedy algorithm converge to a fixed point (derived in equation~\eqref{powerseries} for BDT) in the case of the high backhaul regime (case 1). In the case of the limited capacity backhaul regime (case 2), we can still achieve convergence under BDT as the spectral radius constraint of Theorem~\ref{finiteConvf_net} holds. In the case of the greedy algorithm, however, the transmit powers suffer from an oscillatory trend between either overloading the backhaul links or under-utilizing them. Fig.~\ref{fig:convg}b plots the percentage of time the system takes to reach an equilibrium under BDT within 100 iterations. This figure helps verify Theorem~\ref{finiteConvf_net} using those trials where $\rho\left({\bf{M}}\right)<1$ holds true. We observe that the system always converges under BDT for low values of $\tau$ if $Z$ is closer to one.
\begin{figure*}[!t]
\begin{center}
\subfigure[]{\includegraphics[width=0.48\textwidth]{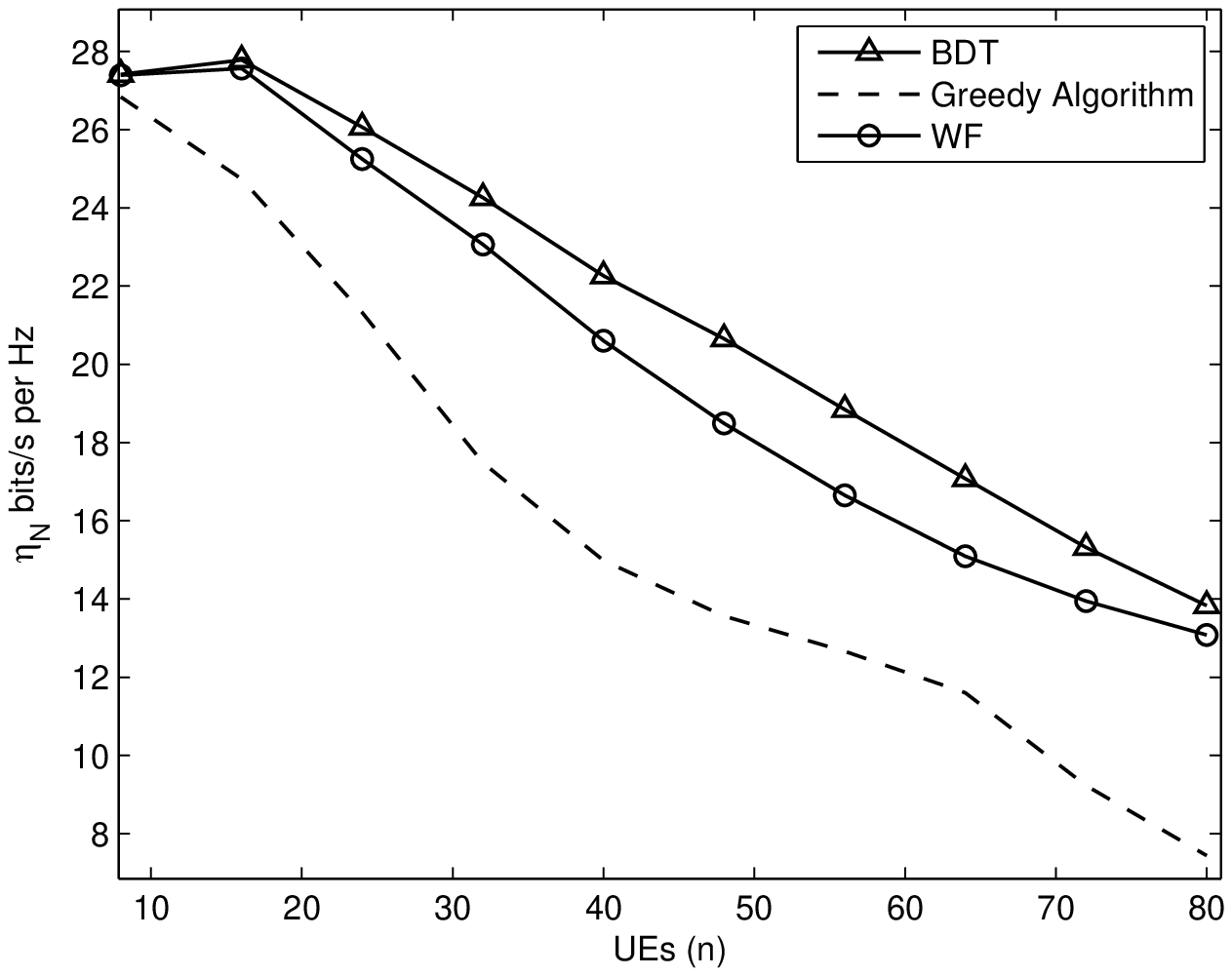}}
\subfigure[]{\includegraphics[width=0.48\textwidth]{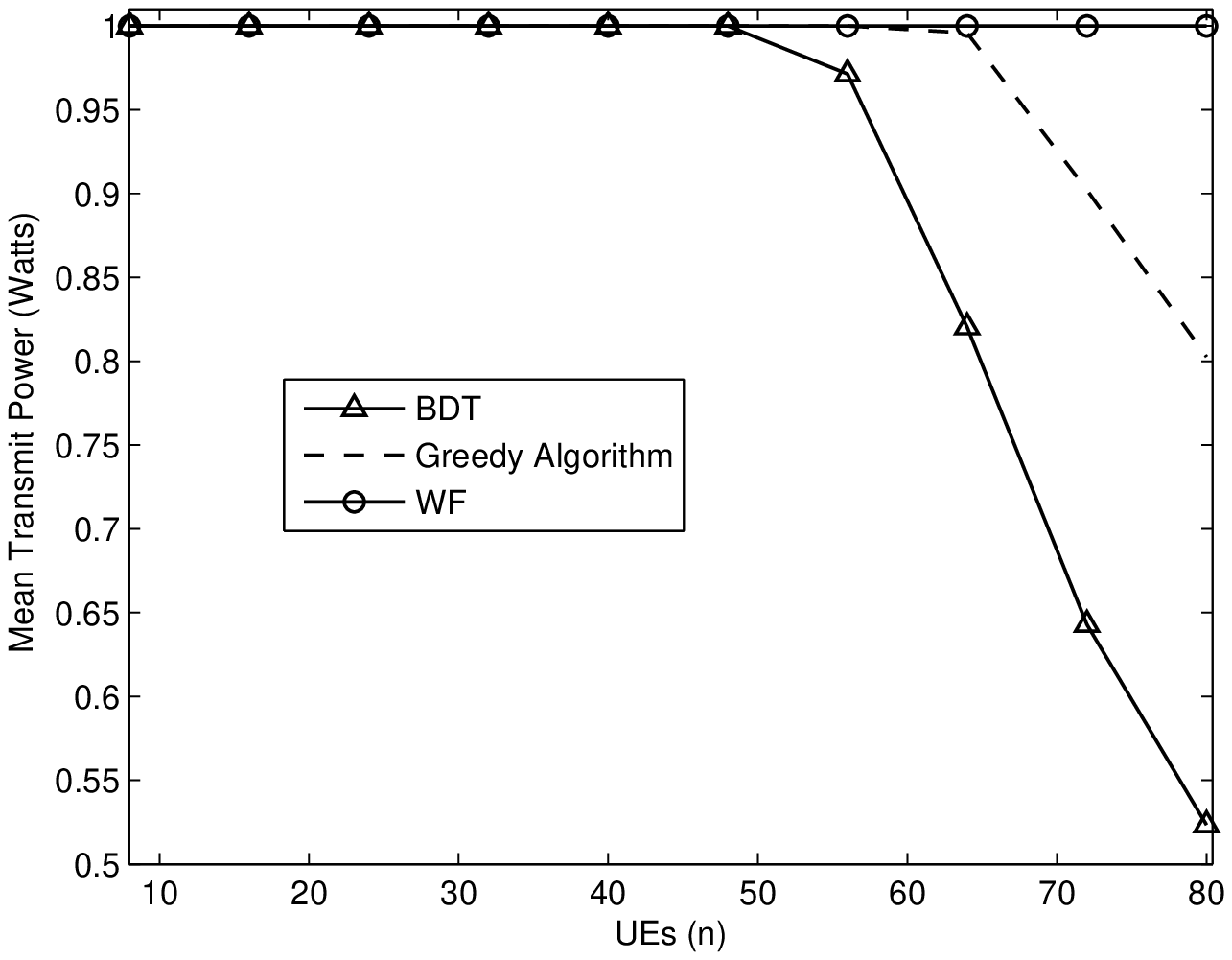}}
\caption{{Impact of network size: BDT offers the best performance over the entire range of $n$: (a) When $n$ is low, BDT yields significant network capacity improvement whereas in (b) when $n$ is large, it affords significant transmit power saving. }}
\label{fig:n}
\end{center}
\end{figure*}    
\begin{figure*}[!t]
\begin{center}
\subfigure[]{\includegraphics[width=0.48\textwidth]{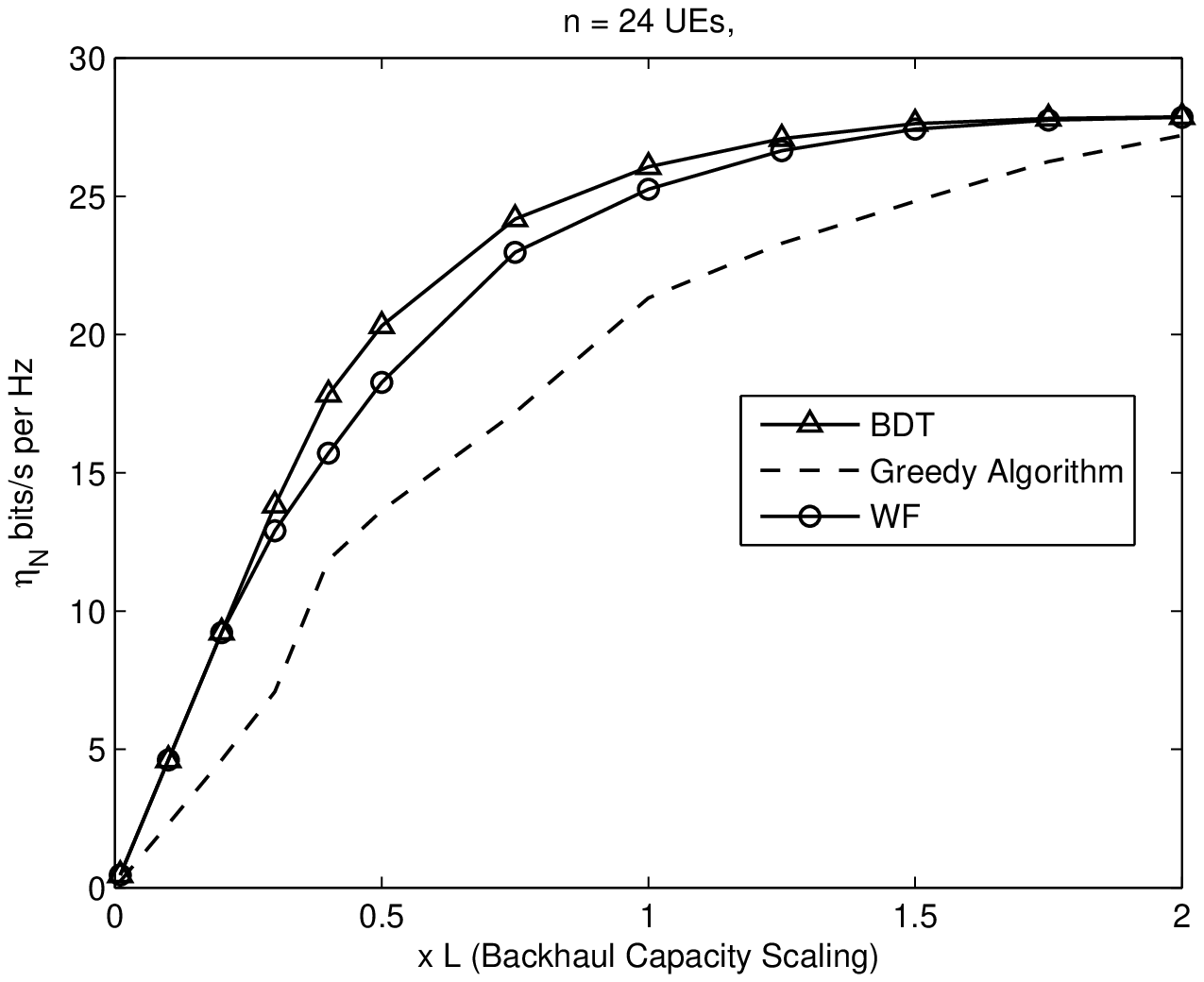}}
\subfigure[]{\includegraphics[width=0.48\textwidth]{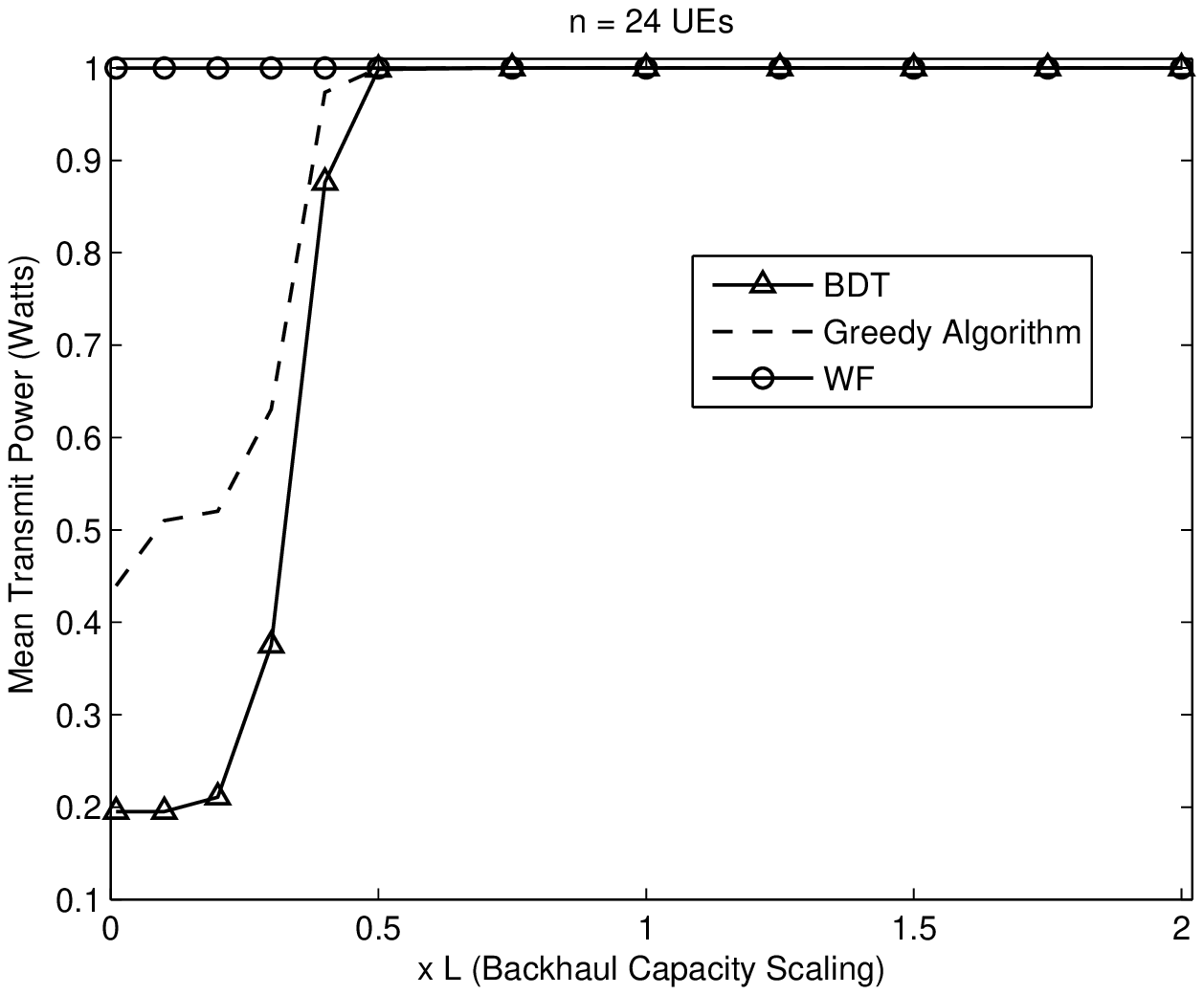}}
\caption{{Effect of varying backhaul capacity: (a) When the backhaul capacities are large, (a) nodes maximize data rate in a stable manner whereas in (b) with limited backhaul capacity, nodes become power-efficient.}}
\label{fig:backhaul}
\end{center}
\end{figure*}
\begin{figure*}[!t]
\begin{center}
\subfigure[]{\includegraphics[width=0.48\textwidth]{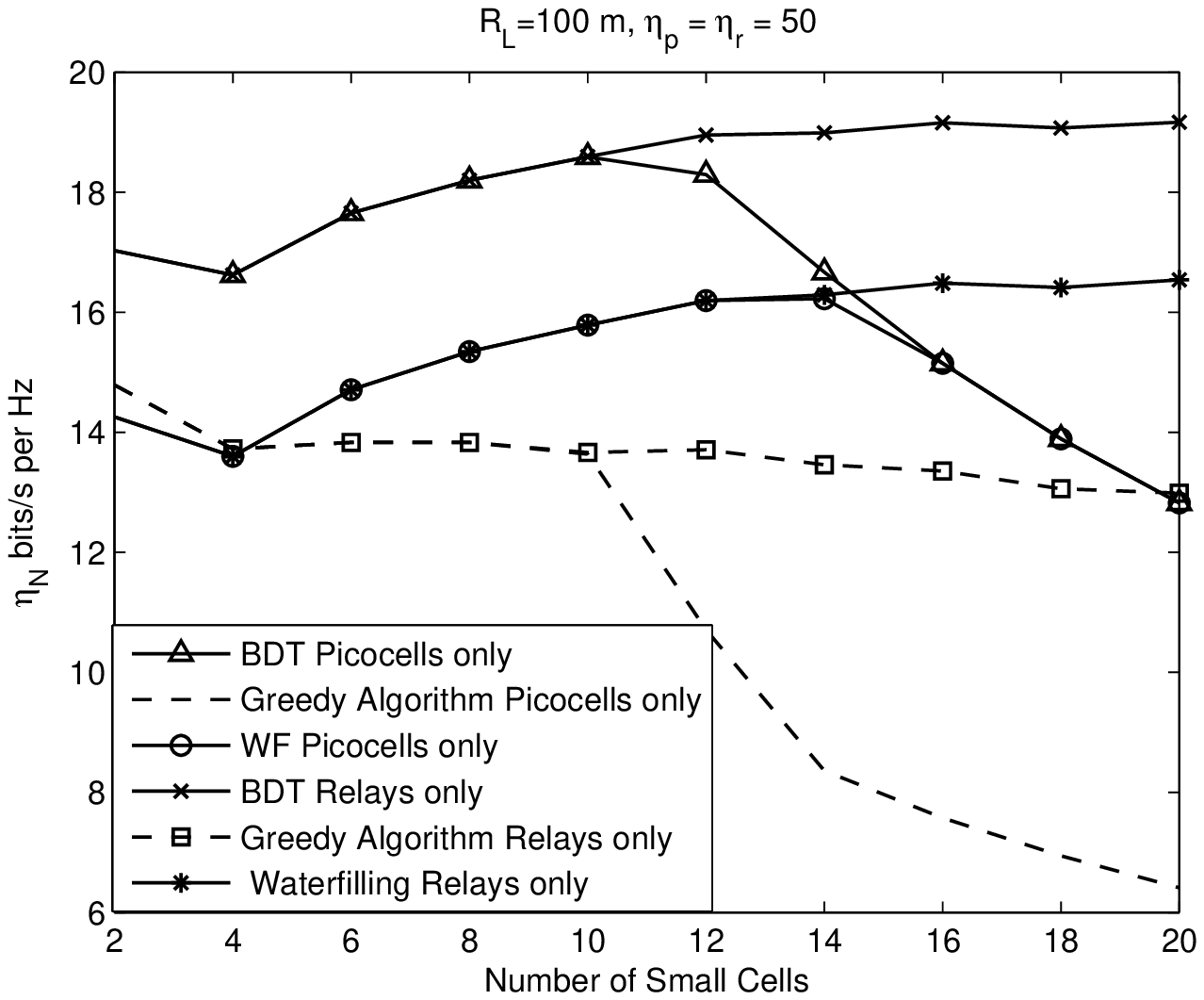}}
\subfigure[]{\includegraphics[width=0.48\textwidth]{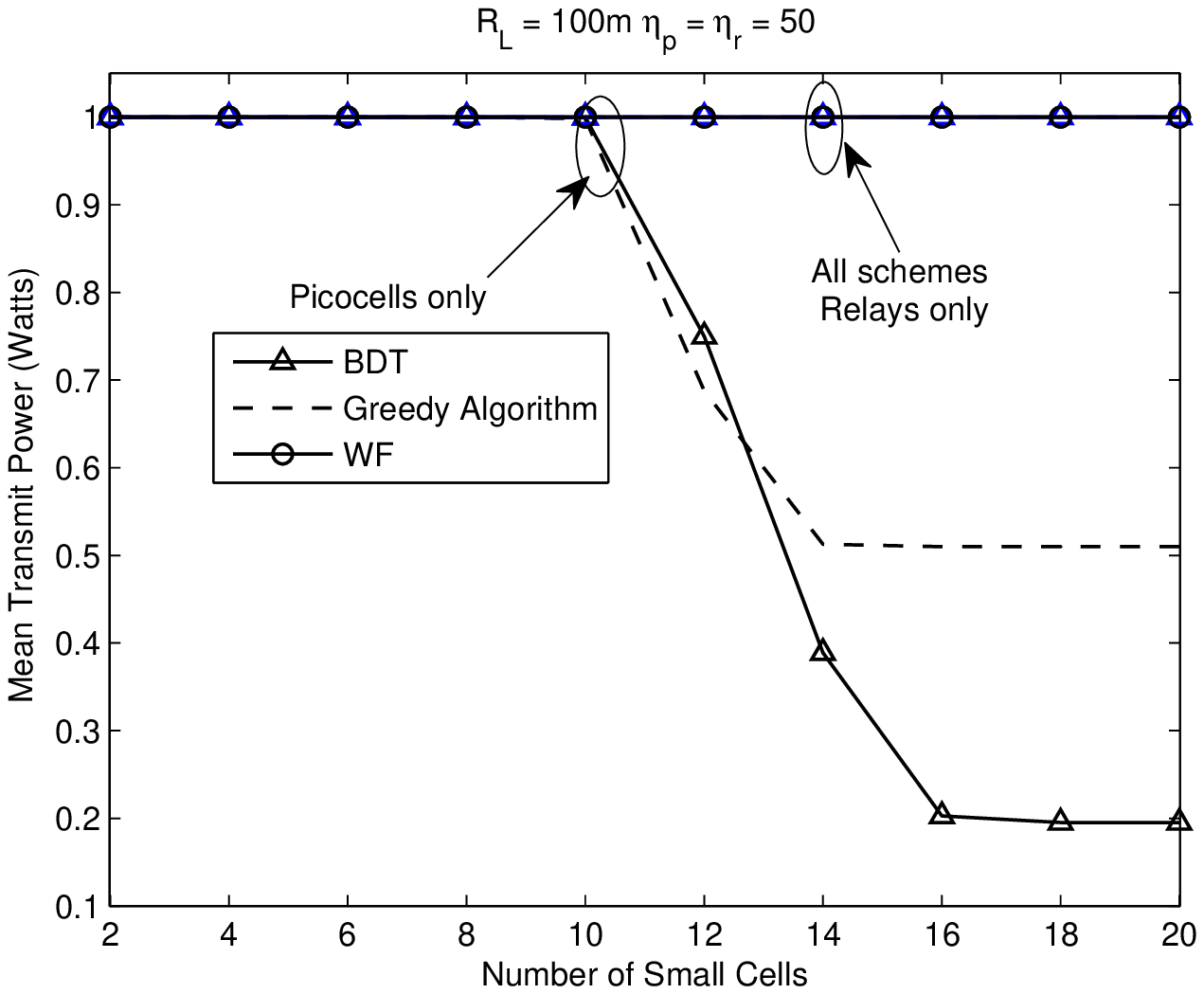}}
\caption{{In above, there are 3 UEs per each small cell. BDT offers increased improvement in terms of (a) $\eta_N$ when the  the number of small cells is fewer whereas, conversely, in (b) it offers better power use especially with increasing number of cells.}}
\label{fig:smallcells}
\end{center}
\end{figure*}    
\subsection{Aggregate end-to-end data rate and power consumption}
Next, we plot network's aggregate end-to-end data rate $\eta_N$ normalized with the total bandwidth used by the system. In Fig.~\ref{fig:n}a, we observe that BDT yields significant data rate improvement over waterfilling and greedy algorithm when $n$ (i.e. number of UEs) is small. As shown in Fig.~\ref{fig:n}b, BDT enables UEs to allocate less transmit power for the same data rate performance when there is increased load on the backhaul links caused by a large value of $n$. Next in Fig.~\ref{fig:backhaul}, we consider the impact of the backhaul capacity on system performance and average transmit power. We adjust the individual backhaul capacities as $\eta_r= 100 \cdot L $, $\eta_p= 200 \cdot L $ Mbps and $\eta_b=1000 \cdot L$ Mbps where $L\in[0-2]$ is a scaling factor for backhaul capacities. We observe from Fig.~\ref{fig:backhaul} that the best data rate performance is achieved under BDT when the backhaul capacities are large (high $L$). When the backhaul capacities are low ($L<0.3$), the UEs use 40 $\%$ transmit power for the same achieved data rates under BDT in comparison with waterfilling where all available transmit power is consumed. In contrast, the greedy algorithm suffers from performance instability due to rapid or oscillatory adaptations by the transmitters and,thus, proves to be a naive strategy. To reiterate, the mechanism under BDT introduces hysteresis into the adaptations and enable nodes to improve their data rate with significantly lower transmit power consumption.

Finally in {Fig.~\ref{fig:smallcells}}, we plot results by varying the number of small cells within the region separated by at least $2R_L$ m, with 3 UEs per cell and $\eta_r=\eta_p=50$ Mbps. We can observe in Fig.~\ref{fig:smallcells}a that the data rate performance eventually declines with increasing number of picocells (and correspondingly the UEs). In contrast, no such decline is observed when increasing number of relays as UEs adaptively send a higher data rate indirectly via the MBS which possesses a high capacity backhaul. Moreover, as shown in Fig.~\ref{fig:smallcells}b the UEs also generally utilize their transmit power more efficiently under BDT.

\section{Conclusion}
We have proposed a transmit power allocation scheme for heterogeneous cellular networks with non-ideal backhaul links where UEs have dual connectivity. We have demonstrated via simulation results that the UEs can significantly improve their energy-efficiency and data rates by taking into account the backhaul load. We have also illustrated that our scheme can achieve convergence in dynamic and complex wireless systems. Future work could consider the impact of cooperation and interference cancellation at the receivers. Moreover, we would also consider generalization of our proposed scheme to a network where each node is connected to more than two PoAs.

\bibliography{PaperBibliography}

\end{document}